\documentclass[%
jmp,
 amsmath,amssymb,
preprint,
onecolumn %%%%% Comment onecolumn to override one column format %%%%%%%
]{revtex4-1}

\usepackage{graphicx}% Include figure files
\usepackage{dcolumn}% Align table columns on decimal point
\usepackage{bm}% bold math
\usepackage[utf8]{inputenc}
\usepackage[T1]{fontenc}
\usepackage{mathptmx}
\usepackage{etoolbox}

\usepackage{colordvi}
\usepackage{amsthm}

\usepackage{hyperref}

\let\q=\quad

\makeatletter
\def\@email#1#2{%
 \endgroup
 \patchcmd{\titleblock@produce}
  {\frontmatter@RRAPformat}
  {\frontmatter@RRAPformat{\produce@RRAP{*#1\href{mailto:#2}{#2}}}\frontmatter@RRAPformat}
  {}{}
}%
\makeatother
\begin{document}

%\preprint{AIP/123-QED}

\title[On the multiplicity of density operator representations]{On the Multiplicity of Density Operator Representations}
% Force line breaks with \\
\author{Gianfranco Cariolaro}
\affiliation{ 
Department of Information Engineering, University of Padova,\\
Via Gradenigo 6/B - 35131 Padova, Italy\\ %\\This line break forced with \textbackslash\textbackslash
}%
\author{Edi Ruffa}%
 \email{edi.ruffa@ieee.org.}
\affiliation{ 
Vimar SpA,\\
Via IV Novembre, 32 - 36063 Vicenza, Italy \\ %\\This line break forced with \textbackslash\textbackslash
}%

%\homepage{http://www.Second.institution.edu/~Charlie.Author.}
%\affiliation{%
%Second institution and/or address%\\This line break forced% with \\
%}%

\date{June 21, 2024}


\begin{abstract}
The density operator is usually defined starting from a set of kets in the Hilbert space and a probability distribution.
From this definition it is easy to obtain a factorization of a given density operator, here called density factor (DF). The multiplicity and the variety
of DFs is investigated using the tools of Matrix Analysis, arriving in particular to establish the DF with minimal size.
The approach based on Matrix Analysis does not seem to be available elsewhere. 
\end{abstract}

\maketitle

\section{\label{sec01}Introduction\protect}

The {\bf density operator}, also called {\bf density matrix}, 
is one of the most important tools in Quantum Mechanics. This operator
was introduced independently by Landau [\onlinecite{Lan27}] and by von Neumann [\onlinecite{von27}]. But  Nobel Prize winner Roy Jay Glauber  contributed most to establishing the importance of the density operator
in numerous articles and books,  since the time of the Manhattan project.
For the rest of the story on density operator see  [\onlinecite{Gla63}], [\onlinecite{Gla68}]. 
 
The density operator is introduced mainly to take into account the thermal noise, but it has a more general role. Also pure states can be represented as a degenerate form of density operators. This  operator is usually defined starting
 from a set of kets in the Hilbert space and a probability distribution.
 From this formulation it is easy to get a factorization of the form
 $\rho= \Psi\;\Psi^*$, where  $\Psi^{*}$ is the conjugate transpose of $\Psi$.
 We find it convenient to call $\Psi$ a {\bf density factor} (DF) of 
 the density operator $\rho$.
 Note that, while a DF  $\Psi$ uniquely determines  $\rho$, for a given $\rho$, one may find infinitely many DFs, also with different sizes.
 This multiplicity has been investigated  by Hugston, Josa and Wootters [\onlinecite{Hug93}], but in this paper it is reconsidered in a new form, completely based on matrix analysis by applying  the singular valued decomposition (SVD) and the eigendecomposition (EID). 
 
 The paper is organized as follows. In Section \ref{sec02} the standard definition of density operator and the definition of density factor are introduced.
 In Sections \ref{sec03} and \ref{sec04} the variety of density factors are considered, establishing in particular the minimum form (orthonormal DF). Throughout the theory is illustrated by specific examples,

\section{\label{sec02}Definitions\protect}

\subsection{\label{sec:level2a}Definition of density operator}

In an $n$-dimensional Hilbert space $\mathcal{H}$ a (discrete) density operator is introduced in the form
\begin{equation}
\rho=\sum_{i=1}^{k} p_{i}\left|\widehat{{\bm{\psi}}}_{i}\right\rangle\left\langle\widehat{{\bm{\psi}}}_{i}\right|
\label{eq01}
\end{equation}
where $\left|\widehat{{\bm{\psi}}}_{i}\right\rangle$ are {\it normalized} quantum states of $\mathcal{H}$ and $\left\{p_{i}\right\}_{i=1,\ldots, k}$ is a  probability
distribution, that is, with $p_{i} \geq 0$ and $\sum_{i=1}^{k} p_{i}=1$. The properties of $p_{i}$ assure that $\rho$ is a positive semidefinite (PSD) operator with $\operatorname{Tr}(\rho)=1$.

In (\ref{eq01}) the density operator $\rho$ is defined starting from the DF of $k$ {\it unnormalized} states
\begin{subequations}
\begin{equation}
\bm{\Psi}=\Biggl[\sqrt{p_{1}}\left| \widehat{\bm{\psi}}_{1} \right\rangle, \ldots, \sqrt{p_{k}} \left| \widehat{\bm{\psi}}_{k} \right\rangle\Biggr] = \Biggl[ \left| \bm{\psi}_{1} \right\rangle, \ldots, \left| \bm{\psi}_{k}\right\rangle \Biggr]
\label{eq02a}
\end{equation}
where
\begin{equation}
\left|\bm{\psi}_{i}\right\rangle = \sqrt{p_{i}} \left| \widehat{{\bm{\psi}}}_{i}\right\rangle, \quad i=1, \ldots, k . 
\label{eq02b}
\end{equation}
\end{subequations}
In fact, the density operator (\ref{eq01}) can be rewritten in matrix form as
\begin{equation}
\rho = \bm{\Psi} \bm{\Psi}^{*} = \left[\left| \bm{\psi}_{1}\right\rangle, \ldots,\left| \bm{\psi}_{k}\right\rangle\right] \left[\begin{array}{c}
\left\langle\widehat{\bm{\psi}}_{1}\right|  \\
\vdots \\
\left\langle\widehat{\bm{\psi}}_{k}\right|
\end{array}\right]
\label{eq03}
\end{equation}
where $\bm{\Psi}^{*}$ is the conjugate transpose of $\bm{\Psi}$.

\subsection{\label{sec:level2b}Definition of Density Factor (DF)}

The matrix $\bm{\Psi}$ may be regarded as a representation  of the given density operator $\rho$, since it contains all the information of $\rho$. It seems to be natural from the factorization $\rho=\bm{\Psi}\,\bm{\Psi}^{*}$ to call it as a {\bf a density factor} associated to the given $\rho$.

Note that, while a DF  $\bm{\Psi}$ uniquely determines a density operator $\rho$, for a given $\rho$, one may find infinitely many DFs, also with different lengths $k$.

The purpose of this paper is to investigate the  multiplicity of the possible density operators. This topic has been investigated in a letter [\onlinecite{Hug93}] by Hugston, Josa and Wootters, but here it is reconsidered in a new form, completely based on matrix analysis. For instance, in [\onlinecite{Hug93}] the compact form (\ref{eq03}) and the consequent application of the singular valued decomposition (SVD) are not considered. In particular, the SVD states a clear link between the decomposition (eigendecomposition) of a density operator and  its DF. Another interpretation of DF, not considered in [\onlinecite{Hug93}], is as a factor of the density operator, as stated by (\ref{eq03}). It is important to remark that the DFs  of density operators play a fundamental role in quantum detection based on the square root measurement [\onlinecite{Eldar4}] and also in other applications.

\newtheorem{exmp}{Example}

%\bigskip

%\paragraph*{Example 1.}

\begin{exmp}

Consider the density operator obtained from the $k=3$ normalized states of $\mathcal{H}=\mathbb{C}^{4}$
$$
\left|\widehat{\bm{\psi}}_{1}\right\rangle = \frac{1}{2}\left[\begin{array}{c}
1  \\
-i \\
-1 \\
i
\end{array}\right],
\quad\left|\widehat{\bm{\psi}}_{2}\right\rangle = \frac{1}{2 \sqrt{6}}
\left[\begin{array}{c}
 2   + \sqrt{2} \\
-2 i + \sqrt{2} \\
-2   + \sqrt{2} \\
 2 i + \sqrt{2}
\end{array}\right],
\quad\left|\widehat{\bm{\psi}}_{3}\right\rangle = \frac{1}{2 \sqrt{6}}
\left[\begin{array}{c}
 2        - \sqrt{2} \\
-2 i      - \sqrt{2} \\
-2        - \sqrt{2} \\
 2 i      - \sqrt{2}
\end{array}\right]
$$

with probabilities $\left\{p_{1},p_2, p_3\right\}=
\left\{\frac14,\frac38, \frac38\right\}$. The expression of $\rho$ is therefore
$$
\rho=
\frac{1}{4} \left|\widehat{\bm{\psi}}_{1}\right\rangle \left\langle \widehat{\bm{\psi}}_{1}\right| + 
\frac{3}{8} \left|\widehat{\bm{\psi}}_{2}\right\rangle \left\langle \widehat{\bm{\psi}}_{2}\right| +
\frac{3}{8} \left| \widehat{\bm{\psi}}_{3}\right\rangle \left\langle \widehat{\bm{\psi}}_{3}\right| = 
\frac{1}{16}
\left[\begin{array}{cccc}
  4       & 1 + i 3 &   -2   &  1 -i 3 \\
  1 - i 3 &    4    & 1 + i3 &    -2   \\
  -2      & 1 - i 3 &   4    & 1  +i 3 \\
  1 + i 3 &    -2   & 1 - i3 &     4 
\end{array}\right]\,.
$$

The unnormalized states $\bm{\psi}_{i}=\sqrt{p_{i}}\; \widehat{\bm{\psi}}_{i}$ are
$$
\left|\bm{\psi}_{1}\right\rangle = \frac{1}{4} 
\left[\begin{array}{c}
 1 \\
-i \\
-1 \\
 i
\end{array}\right], \quad\left| \bm{\psi}_{2}\right\rangle = \frac{1}{8} \left[\begin{array}{c}
 2   + \sqrt{2} \\
-2 i + \sqrt{2} \\
-2   + \sqrt{2} \\
 2 i + \sqrt{2}
\end{array}\right], \quad\left| \bm{\psi}_{3}\right\rangle = \frac{1}{8} \left[\begin{array}{c}
 2   - \sqrt{2} \\
-2 i - \sqrt{2} \\
-2   - \sqrt{2} \\
 2 i - \sqrt{2}
\end{array}\right]
$$

and the corresponding factor is
\begin{equation}
\bm{\Psi} = \left[\left|\bm{\psi}_{1}\right\rangle, \left|\bm{\psi}_{2}\right\rangle, \left|\bm{\psi}_{3}\right\rangle\right] = \frac{1}{4} 
\left[\begin{array}{ccc}
 1 &  1+\frac{1}{\sqrt{2}} &  1-\frac{1}{\sqrt{2}} \\
-i & -i+\frac{1}{\sqrt{2}} & -i-\frac{1}{\sqrt{2}} \\
-1 & -1+\frac{1}{\sqrt{2}} & -1-\frac{1}{\sqrt{2}} \\
 i &  i+\frac{1}{\sqrt{2}} &  i-\frac{1}{\sqrt{2}}
\end{array}\right]\,.
\label{eq04}
\end{equation}
We can check that
$$
\bm{\Psi} \bm{\Psi}^{*} = \left[\left| \bm{\psi}_{1}\right\rangle, \left| \bm{\psi}_{2}\right\rangle, \left| \bm{\psi}_{3}\right\rangle \right]\left[ \begin{array}{l}
\left\langle \bm{\psi}_{1} \right| \\
\left\langle \bm{\psi}_{2} \right| \\
\left\langle \bm{\psi}_{3} \right|
\end{array}\right]=\rho\,.
$$
Note also that
$$
\bm{\Psi}^{*} \bm{\Psi} = \left[\begin{array}{ccc}
 \frac{1}{4} & -\frac{1}{4} & -\frac{1}{4} \\
-\frac{1}{4} &  \frac{3}{8} &  \frac{1}{8} \\
-\frac{1}{4} &  \frac{1}{8} &  \frac{3}{8}
\end{array}\right]\neq \bm{\Psi} \bm{\Psi}^{*}
$$
which shows that the states $\left| \bm{\psi}_{i} \right\rangle$ are not orthonormal.

\end{exmp}

\subsection{\label{sec:level2c}The class of density factors of a given density operator}

First we introduce a few definitions. A DF $\bm{\Psi}$ of $\rho$ with $k$ components will be called a {\bf $k$--density factor} ($k$-DF)  of $\rho$. It is easy to see that the minimum value of $k$ is given by the rank $r$ of $\rho$, which is also the rank of any DF of $\rho$, but the value of $k$ may be arbitrarily large. An $r$-DF , with $r=\operatorname{rank}(\rho)$, will be called {\bf  minimum density factor} of $\rho$. A DF $\left.\bm{\Psi}=\left[ \left| \bm{\psi}_{1} \right\rangle|, \ldots, | \bm{\psi}_{k}\right\rangle \right]$, where the states are orthonormal, will be called an {\bf orthonormal density factor} of $\rho$. For a $k$--DF $\bm{\Psi}$ the orthonormality condition can be written in the form
\begin{equation}
\bm{\Psi}^{*} \bm{\Psi} = \operatorname{diag} \left\{\left\langle \bm{\psi}_{1} \mid \bm{\psi}_{1} \right\rangle, \ldots, \left\langle \bm{\psi}_{k} \mid \bm{\psi}_{k} \right\rangle\right\}
\label{eq05}
\end{equation}
where $\left\langle \bm{\psi}_{i} \mid \bm{\psi}_{i} \right\rangle$ gives the probabilities.

Note that an orthonormal $k$-DF  $\bm{\Psi}$ is necessarily minimum. In fact, the $k$ orthonormal columns of $\bm{\Psi}$ are linearly independent and therefore $\bm{\Psi}$ has $\operatorname{rank} k$, but $\operatorname{rank}(\bm{\Psi})=\operatorname{rank}(\rho)=r$.

In the previous example, where $k=3$, the rank of $\rho$ is $r=2$ and $\bm{\Psi}^{*} \bm{\Psi}$ is not a diagonal matrix, so $\bm{\Psi}$ is neither minimum nor orthonormal.

\subsection{\label{sec:level2d}Minimum density factor from the EID of $\rho$}

We now show that it is easy to find a minimum DF  of $\rho$.

\newtheorem{prop}{Proposition}

\begin{prop}

\label{prop01}

Let $\rho$ be a density operator in an $n$-dimensional Hilbert space $\mathcal{H}$ and let $r=\operatorname{rank}(\rho)$. Next, consider the eigendecomposition (EID) of $\rho$

\begin{equation}
\rho=\sum_{i=1}^{r} \sigma_{i}^{2}\left| \widehat{\bm{u}}_{i} \right\rangle\left
\langle \widehat{\bm{u}}_{i} \right|= \widehat{\bm{U}}  \bm{\Sigma}^{2}  \widehat{\bm{U}}^{*}
\label{eq06}
\end{equation}
where $\sigma_{i}^{2}$ are the $r$ positive eigenvalues of $\rho,\left|\widehat{\bm{u}}_{i} \right\rangle$ are the corresponding orthonormal eigenvectors, $ \widehat{\bm{U}} = \left[\left| \widehat{\bm{u}}_{1} \right\rangle, \ldots, \left| \widehat{\bm{u}}_{r} \right\rangle \right]$ and $\bm{\Sigma}^{2}=\operatorname{diag}\left\{\sigma_{1}^{2}, \ldots, \sigma_{r}^{2}\right\}$. Then,
\begin{equation}
\bm{U} = \widehat{\bm{U}}  \bm{\Sigma} = \left[\bm{u}_{1}, \ldots, \bm{u}_{r} \right], \q \text{with $\bm{u}_i = \sigma_i \, \widehat{\bm{u}}_i$} 
\label{eq07}
\end{equation}
is a minimum orthonormal DF of $\rho$.
\end{prop}

\begin{proof}
%\noindent{\bf Proof} \q 
Clearly $\bm{U} \bm{U}^{*}=\rho$, that is, $\bm{U}$ is a factor of $\rho$. Since $\rho$ is Hermitian and PSD, its non zero eigenvalues are positive, as remarked by their symbols $\left(\sigma_{i}^{2}\right)$. Moreover, $\operatorname{Tr}(\rho) = 1$ and the trace is given by the sum of the eigenvalues. Hence $\sigma_{i}^{2}$ are normalized probabilities. Finally note that in (\ref{eq06}) $ \widehat{\bm{U}} $ collects $r$ orthonormal eigenvectors, so $\bm{U}$ consists of orthonormal kets and represents a minimum orthonormal DF  of $\rho$.
%\bigskip
\end{proof}

%\noindent{\bf Example 2} \q 

\begin{exmp}

Reconsider the density operator of the previous example, which has rank $r=2$. The EID is given $\rho= \widehat{\bm{U}}  \bm{\Sigma}^{2}  \widehat{\bm{U}}^{*}$ with

$$
\widehat{\bm{U}} =\frac{1}{2}
\left[\begin{array}{cc}
- 1 & - 1 \\
  i & - 1 \\
  1 & - 1 \\
- 1 & - 1
\end{array}\right], \quad \bm{\Sigma}^{2} = \frac{1}{4} 
\left[\begin{array}{cc}
  3 & 0 \\
  0 & 1
\end{array}\right] \quad \rightarrow \quad \bm{\Sigma} = \frac{1}{2} 
\left[\begin{array}{cc}
\sqrt{3} & 0 \\
  0      & 1
\end{array}\right]\,.
$$
Hence,
\begin{equation}
\bm{\Psi}_{0}= \widehat{\bm{U}}  \bm{\Sigma} = \frac{1}{4} 
\left[\begin{array}{cc}
- \sqrt{3} & -1 \\
i \sqrt{3} & -1 \\
  \sqrt{3} & -1 \\
- \sqrt{3} & -1 
\end{array}\right]
\label{eq08}
\end{equation}
is a minimum orthonormal DF  of $\rho$. We can check that
$$
\widehat{\bm{U}}^{*}  \widehat{\bm{U}} = 
\left[\begin{array}{ll}
 1 & 0 \\
 0 & 1
\end{array}\right], \quad \bm{U} \bm{U}^{*}=\rho\,.
$$

\end{exmp}

\subsection{Minimum DF from an arbitrary DF}

Given an arbitrary DF  $\bm{\Psi}$ of $\rho$ it is possible to find directly a minimum orthonormal DF.
%\bigskip
 
%\noindent{\bf Proposition 2} \q 

\begin{prop}

\label{prop02}

Let $\bm{\Psi}$ be an arbitrary $k$-DF  of $\rho$. Consider the singular-value decomposition (SVD) [\onlinecite{Horn98}] of the $n \times k$ matrix $\bm{\Psi}$, given by
\begin{equation}
\bm{\Psi} = \sum_{i=1}^{r} \sigma_{i} \left| \widehat{\bm{u}}_{i} \right\rangle \left\langle \widehat{\bm{v}}_{i}\right| = \widehat{\bm{U}}  \bm{\Sigma} \widehat{\bm{V}}^{*} 
\label{eq09}
\end{equation}
where $\sigma_{i}$ are the square roots of the $r$ positive eigenvalues $\sigma_{i}^{2}$ of $\bm{\Psi} \bm{\Psi}^{*}=\rho, \bm{\Sigma}=\operatorname{diag}\left\{\sigma_{1}, \ldots, \sigma_{r}\right\}$, $\left|\widehat{\bm{u}}_{i}\right\rangle$ and $ \widehat{\bm{U}} $ are the same as in the EID of (\ref{eq06}), $\left| \widehat{\bm{v}}_{i}\right\rangle$ are orthonormal vectors of length $k$ and $\widehat{V} = \left[\left| \widehat{\bm{v}}_{1}\right\rangle, \ldots, \left| \widehat{\bm{v}}_{r}\right\rangle\right]$. Then, a minimum orthonormal DF  of $\rho$ is given by
\begin{equation}
\bm{\Psi}_{0} = \bm{U} = \widehat{\bm{U}}  \bm{\Sigma}\,.
\label{eq10}
\end{equation}

\end{prop}

%\bigskip

%\noindent{\bf Proof} 

\begin{proof}
The expression of $\bm{\Psi}_{0}$ is the same as in Proposition \ref{prop01}. Note that from (\ref{eq09}), considering that $\bm{V}^{*} \bm{V} = I_{r}$ (the $r \times r$ identity matrix), we find
$$
\bm{\Psi} \bm{\Psi}^{*}= \widehat{\bm{U}}  \bm{\Sigma} \widehat{\bm{V}}^{*} \widehat{\bm{V}} \bm{\Sigma}  \widehat{\bm{U}}^{*} = \widehat{\bm{U}}  \bm{\Sigma} \bm{\Sigma}  \widehat{\bm{U}}^{*} = \bm{\Psi}_{0} \bm{\Psi}_{0}^{*} = \rho
$$
that is the EID of $\rho$.
\end{proof}

%\bigskip

%\noindent{\bf Example 3} \q 

\begin{exmp}

The DF  given by (\ref{eq04}) has rank $r=2$ and its SVD is $\bm{\Psi} = \widehat{\bm{U}}  \bm{\Sigma} \widehat{\bm{V}}^{*}$ with

$$
\widehat{\bm{\bm{U}}} = \frac{1}{2}
\left[\begin{array}{cc}
 - 1 & - 1 \\
   i & - 1 \\
   1 & - 1 \\
 - 1 & - 1 
\end{array}\right], \quad \widehat{\bm{V}} = \frac{1}{\sqrt{6}}
\left[\begin{array}{cc}
 \sqrt{2} & 0        \\
 \sqrt{2} & \sqrt{3} \\
-\sqrt{2} & \sqrt{3}
\end{array}\right], \quad \bm{\Sigma} = \frac{1}{2} \left[
\begin{array}{cc}
 \sqrt{3} & 0 \\
        0 & 1
\end{array}\right]\,.
$$

Hence, we find the minimum factor $\bm{\Psi}_{0}$ obtained with the EID of $\rho$.

\end{exmp}

\section{\label{sec03}Other ways to modify a density factor}

Let $\bm{\Psi}$ be a $k$-DF  of $\rho$, that is, $\bm{\Psi} \bm{\Psi}^{*}=\rho$, and let $\bm{A}$ be an arbitrary $k \times  p$ complex matrix that verifies the condition
\begin{equation}
\bm{A} \bm{A}^{*} = \bm{I}_{k}
\label{eq11}
\end{equation}
where $\bm{I}_{k}$ is the identity matrix of order $k$. Then,
\begin{equation}
\bm{\Phi} = \bm{\Psi} \bm{A}
\label{eq12}
\end{equation}
is a $p$-DF  of $\rho$. In fact,
$$
\bm{\Phi} \bm{\Phi}^{*}=\bm{\Psi} \bm{A} \bm{A}^{*} \bm{\Psi}^{*}=\bm{\Psi} \bm{\Psi}^{*} = \rho.
$$
Note that $k \geq r = \operatorname{rank}(\rho)=\operatorname{rank}(\bm{\Psi})$. Condition (\ref{eq11}) states that the $k$ columns $\bm{a}_{i}$ of $\bm{A}$ are orthonormal vectors, that is, they verify the condition $\bm{a}_{i} \bm{a}_{j}^{*} = \delta_{i j}$. Considering that $\bm{a}_{i} \in \mathbb{C}^{p}$, we have the condition $p \geq k$, because we cannot find in $\mathbb{C}^{p}$ an orthonormal set $\left\{\bm{a}_{1}, \ldots, \bm{a}_{k}\right\}$ with more than $p$ components.

\bigskip

Now we relate an arbitrary DF  of $\rho$ to a minimum orthonormal DF.

%\noindent{\bf Proposition 3} \q 

\begin{prop}
\label{prop03}

An arbitrary $k$-DF  $\bm{\Phi}$ of $\rho$ is related to a reference minimum orthonormal DF  $\Psi_{0}$ in the form
\begin{subequations}
\begin{equation}
\bm{\Phi} = \bm{\Psi}_{0} \bm{A}_{0}
\label{eq13a}
\end{equation}
where $\bm{A}_{0}$ is an $r \times k$ matrix given by
\begin{equation}
\bm{A}_{0} = \bm{\Sigma}^{-2} \bm{\Psi}_{0}^{*} \bm{\Phi} 
\label{eq13b}
\end{equation}
\end{subequations}
$\bm{\Sigma}^{2}$ being the diagonal matrix formed by the positive eigenvalues of $\rho$. The matrix $\bm{A}_{0}$ verifies always the condition $\bm{A}_{0} \bm{A}_{0}^{*} = \bm{I}_{r}$.

\end{prop}

%\bigskip
%\noindent{\bf Proof} \q 

\begin{proof}

Considering that a minimum DF  has the form (see (\ref{eq07})) $\bm{\Psi}_{0}= \widehat{\bm{U}}  \bm{\Sigma}$, by left multiplying (\ref{eq13a}) by $ \widehat{\bm{U}}^{*}$, we get
$$
\widehat{\bm{U}}^{*} \bm{\Phi} = \widehat{\bm{U}}^{*} \bm{\Psi}_{0} \bm{A}_{0}= \widehat{\bm{U}}^{*} \widehat{\bm{U}} \bm{\Sigma} \bm{A}_{0} = \bm{\Sigma} \bm{A}_{0}
$$
and, considering that $\bm{\Psi}_{0}= \widehat{\bm{U}} \bm{\Sigma}$, (\ref{eq13b}) follows. Next, from (\ref{eq13b}), considering that $\bm{\Phi} \bm{\Phi}^{*} = \rho$ and also $\bm{\Psi}_{0} \bm{\Psi}_{0}^{*} = \rho$, we get
$$
\bm{A}_{0} \bm{A}_{0}^{*} = \bm{\Sigma}^{-2} \bm{\Psi}_{0}^{*} \bm{\Phi} \bm{\Phi}^{*} \bm{\Psi}_{0} \bm{\Sigma}^{-2} = \bm{\Sigma}^{-2} \bm{\Psi}_{0}^{*} \bm{\Psi}_{0} \bm{\Psi}_{0}^{*} \bm{\Psi}_{0} \bm{\Sigma}^{-2}
$$
where $\bm{\Psi}_{0}^{*} \bm{\Psi}_{0} = \bm{\Sigma}^{2}$. Hence, $\bm{A} \bm{A}^{*} = \bm{I}_{r}$.

\end{proof}

%\bigskip 

%\noindent{\bf Example 4} \q 

\begin{exmp}

Consider the minimum DF  given by (\ref{eq08})

\begin{equation}
\bm{\Psi}_{0}= \widehat{\bm{U}}  \bm{\Sigma} = \frac{1}{4}
\left[\begin{array}{cc}
 -\sqrt{3} & - 1  \tag{14}\\
i \sqrt{3} & - 1 \\
  \sqrt{3} & - 1 \\
- \sqrt{3} & - 1
\end{array}\right]\,.
\label{eq14}
\end{equation}
Using matrix
$$
\bm{A}_{0} = \frac{1}{\sqrt{6}} 
\left[\begin{array}{ccc}
 0        &  \sqrt{3} & -\sqrt{3} \\
-\sqrt{2} &  \sqrt{2} &  \sqrt{2}
\end{array}\right]
$$
which verifies the condition $\bm{A}_{0} \bm{A}_{0}^{*} = \bm{I}_{2}$, we obtain the 3-DF 
$$
\bm{\Phi} = \bm{\Psi}_{0} \bm{A} = -\frac{1}{8 \sqrt{3}} 
\left[\begin{array}{ccc}
  2 & 2 + 3   \sqrt{2} &  2 - 3   \sqrt{2} \\
  2 & 2 - 3 i \sqrt{2} &  2 + 3 i \sqrt{2} \\
  2 & 2 - 3   \sqrt{2} &  2 + 3   \sqrt{2} \\
  2 & 2 + 3   \sqrt{2} &  2 - 3   \sqrt{2}
\end{array}\right]\,.
$$

Another example of DF  modification is obtained with the matrix

$$
\bm{A} = \frac{1}{2\sqrt{2}}
\left[\begin{array}{cccccccc}
   1  &         1             &    1  &        1               &  1   &  1                    &  1  &       1           \\
   1  & e^{-\frac{i \pi}{4}}  &  - i  & e^{-\frac{i 3 \pi}{4}} & - 1  & e^{\frac{i 3 \pi}{4}} &  i  & e^{\frac{i \pi}{4}} \\
   1  & -i                    &  - 1  &        i               &  1   & -i                    & -1  &       i
\end{array}\right]
$$
which verifies the condition $\bm{A} \bm{A}^{*}=\bm{I}_3$.
It gives the 8-DF 

$$
\bm{\Psi}_{0} A = \frac{1}{2 \sqrt{2}}
\left[\begin{array}{cccccccc}
 1 &  1                   & 1  & 1                      &  1  & 1                     & 1  & 1                   \\
 1 & e^{-\frac{i \pi}{4}} & -i & e^{-\frac{i 3 \pi}{4}} &  -1 & e^{\frac{i 3 \pi}{4}} & i  & e^{\frac{i \pi}{4}} \\
 1 & -i                   & -1 & i                      &  1  & -i                    & -1 & i
\end{array}\right]\,.
$$

\end{exmp}

\section{\label{sec04}Conclusions}
We have investigated the multiplicity of density factors (DFs) of a given density operator. The main result is given by Proposition \ref{prop03}, which states that, starting from a minimum orthonormal DF  $\bm{\Psi}_{0}$, of dimension $n \times r$, one can generate all the possible DFs of a given density operator in the form $\bm{\Psi}=\bm{\Psi}_{0} \bm{A}_{0}$, where $\bm{A}_{0}$ is an arbitrary $k \times r$ matrix with orthonormal rows, that is, with $\bm{A}_{0} \bm{A}_{0}^{*}=\bm{I}_{r}$. Note that $k \geq r$ may be arbitrarily large.

In the paper we have considered density operators of {\bf finite order}, as considered by Hugston {\it et al.} in [\onlinecite{Hug93}]. But density operators may have an infinite order, like Gaussian density operators. Then the factorization is possible using the Block espansion [\onlinecite{Car15}].

\nocite{*}
%\section*{References}
\bibliography{multiplicity}% Produces the bibliography via BibTeX.
\end{document}